\RequirePackage{fix-cm}
\documentclass{svjour3}                    % onecolumn (standard format)
\smartqed  % flush right qed marks, e.g. at end of proof

\usepackage[ruled,linesnumbered, noline]{algorithm2e}
\usepackage{amsmath,amssymb,amsfonts}
\usepackage{graphicx}
\usepackage[misc]{ifsym}
\usepackage{booktabs} % For formal tables
\usepackage[misc]{ifsym}
\usepackage{amsmath}
\usepackage{amsfonts}
\usepackage{multirow}
\usepackage{hhline}
\usepackage{color}
\usepackage{pdfpages}
\usepackage{graphics}
\usepackage{subfigure}

\usepackage{graphicx}

\def \opt{{\mathsf{opt}}}

\def\SMK{\mathsf{SMK}}

\def\LA{\mathsf{LA}}

\def\DLA{\mathsf{DLA}}
\def\EDL{\mathsf{EDL}}

\def\NSMK{\mathsf{\mbox{non-monotone } SMK}}

\begin{document}
	
	\title{Enhanced Deterministic Approximation Algorithm  for Non-monotone Submodular Maximization under Knapsack Constraint with Linear Query Complexity}
	%\subtitle{Fast Bicriteria Approximation Algorithm for Minimum Cost Submodular Cover Problem}
	
	\titlerunning{Enhanced Deterministic Approximation Algorithm  for $\SMK$...}        % if too long for running head
	
	\author{Canh V. Pham}
	
	\authorrunning{Canh V. Pham} % if too long for running head
	
	\institute{Canh V. Pham (Corresponding Author) \at
		ORLab, Faculty of Computer Science, Phenikaa University, Hanoi, Vietnam
		\email{canh.phamvan@phenikaa-uni.edu.vn}
	}
	
	\date{Received: date / Accepted: date}

	\maketitle
	
	\begin{abstract}
	In this work, we consider the Submodular Maximization under Knapsack ($\SMK$) constraint problem over the ground set of size $n$. The problem recently attracted a lot of attention due to its applications in various domains of combination optimization, artificial intelligence, and machine learning. We improve the approximation factor of the fastest deterministic algorithm from $6+\epsilon$ to $5+\epsilon$ while keeping the best query complexity of $O(n)$, where $\epsilon >0$  is a constant parameter. Our technique is based on optimizing the performance of two components: the threshold greedy subroutine and the building of two disjoint sets as candidate solutions. Besides, by carefully analyzing the cost of candidate solutions, we obtain a tighter approximation factor.
		%\subclass{MSC code1 \and MSC code2 \and more}
	\end{abstract}
	\section{Introduction}
	Submodular Maximization under a Knapsack ($\SMK$) constraint plays a crucial role in the fields of combinatorial optimization, artificial intelligence, and machine learning. In this problem, it's given a ground set $V$ of sized $n$ and a non-negative (not necessarily monotone) submodular set function  $f: 2^V \mapsto \mathbb{R}_+$. Each element $e \in V$  is assigned a
	positive cost $c(e)$ and there is a budget $B$,  $\SMK$ asks for finding $S\subseteq V$ with minimal total cost $c(S)=\sum_{e \in S}c(e)\leq B$ so that maximizes $f(S)$. 
	
	$\SMK$ problem finds a wide-rage of applications such as maximum weighted cut~\cite{Amanatidis_sampleGreedy,Han2021_knap}, data summarization~\cite{Han2021_knap,fast_icml}, revenue maximization~\cite{Han2021_knap,Cui-streaming21}, information propagation in social networks \cite{k-sub-Canh-csonet21,cor-24} and recommendation systems~\cite{Amanatidis2021a_knap,Amanatidis_sampleGreedy}, thereby it  has paid a lot of attention recently \cite{fast_icml,Amanatidis2021a_knap,Han2021_knap,sub_knap_orlet,li-linear-knap-nip22,Ene_HuyLNg_knap,Lee_nonmono_matroid_knap,Amanatidis_sampleGreedy,Gupta_nonmono_constrained_submax}.

	In addition to focusing on approximation algorithms with tight factors for $\SMK$, people focus on solving the problem within a reasonable time cost, especially in the era of big data. Since query complexity evaluated by the number of required queries to the objective function dominates the running time of an algorithm, it is important to reduce the query complexity of algorithms.	
	Besides, recent researches show that deterministic algorithms often give unique and better solutions than randomized algorithms in practices~\cite{li-linear-knap-nip22,Kuhnle19-nip,pham-ijcai23,kdd-chen23}. 
	Therefore, it is necessary to design efficient deterministic algorithms that both guarantee the theoretical bounds and waste a low query complexity. To the best of our knowledge, the fastest deterministic approximation algorithm is due to Pham {\em et al.}~\cite{pham-ijcai23}. Their algorithm provides an approximation factor of $6+\epsilon$ in linear query complexity of $O(n \log(1/\epsilon)/\epsilon)$ (See the Table~\ref{tab:1} for an overview of deterministic algorithms). However, there is still a large gap between  Pham {\em et al.}'s algorithm and the best factor of $2.6$ in \cite{BuchbinderF19-bestfactor}. This raises an open and interesting question: \textit{Can we improve the approximation factor of an algorithm in linear query complexity for the studied problem?} 
	
	\paragraph{\textbf{Our contributions.}}  In this work, we address the above question by introducing an approximation algorithm with a better factor of $5+\epsilon$ in linear query complexity. Our technique improves the algorithm framework of $\DLA$ in \cite{pham-ijcai23} with a tighter theoretical analysis. Firstly, we re-design the combination of the threshold greedy subroutine and the building of two disjoint solutions. Secondly, we explore a strong connection between the cost of candidate solutions and the optimal cost. Therefore, by more carefully analyzing the above relation, we obtain a tighter approximation factor without increasing computation cost.
	\begin{table*}[h]
		\centering
		\begin{tabular}{ccccc}
			\hline
			\textbf{Reference} 
			& \textbf{Approximation factor} & \textbf{Query Complexity}
			\\
			\hline
			GREEDY	\cite{Gupta_nonmono_constrained_submax} &$6$& $O(n^5)$
			\\
			Twin Greedy \cite{best-dla} & $4+\epsilon$ & $O(n^3\log(n)/\epsilon)$
			\\
			SMKDETACC~\cite{Han2021_knap}  &$6+\epsilon$&$O(n\log(k/\epsilon)/\epsilon)$
			\\
			$\DLA$~\cite{pham-ijcai23}& $6+\epsilon$  & {\boldmath $O(n \log(1/\epsilon)/\epsilon)$}
			\\
			Our algorithm ($\EDL$) & {\boldmath $5+\epsilon$} & {\boldmath $O(n \log(1/\epsilon)/\epsilon)$}
			\\
			\hline
		\end{tabular}
		\caption{Deterministic approximation algorithms for $\NSMK$ problem. $k$ is the maximum cardinality of any feasible solution to $\SMK$.  The best result(s) are bold}
		\label{tab:1}
	\end{table*}
	\paragraph{Paper Organization.} The rest of the paper is structured as follows. 
	Section \ref{sec:relatedwork}	provides the literature review on the $\NSMK$ problem. Preliminaries on
	submodularity and the studied problem are presented in Section \ref{sec:preli}. Section \ref{sec:algs} introduces two proposed algorithms and theoretical analysis. Finally, we conclude this work in Section \ref{sec:con}.
	%%% Checkhere
	\section{Related Works}
	\label{sec:relatedwork}
	% In this section, we review the related work for the $\SMK$ problem.
	%\mt{This section, let's just focus on the non-monotone SMK problem. Move all others to the appendix}
	\paragraph{$\SMK$ problem with monotone function.} Wolsey {\em et al.}
	\cite{sub_Wolsey_knap} first solved the $\SMK$ problem with monotone objective function. They proved that it was NP-hard to give an approximation algorithm the factor of $(e/(e-1))$. The latter works  \cite{sub_knap_orlet,leskovec_celf}  proposed greedy versions with the same optimal approximation factor of $e/(e-1)$ for monotone $\SMK$ problem; however, they took an expensive query complexity of $O(n^5)$. 
	The recent work \cite{mono-smk-o3} kept the optimal factor while reducing query complexity to $O(n^3 \log n)$. Since then, several works have been made to reduce the query complexity to achieve the optimal approximation factor. Authors~\cite{fast_sub} introduced a faster  algorithm that had $O(n^2(\log(n)/\epsilon)^{1/\epsilon^8})$ but this work contained trivial errors~\cite{Ene_HuyLNg_knap}. \cite{Ene_HuyLNg_knap} tried to obtain $(e/(e-1)+\epsilon)$ factor but this work had to handle complicated multi-linear extensions which required an expensive  $O((1/\epsilon)^{O(1/\epsilon^{4})}n\log^2 n)$ number of function evaluations. Another work claimed the factor of $2+\epsilon$ for monotone $\SMK$ within $O(nk)$ \cite{Grigory}. However, a loophole in the theoretical analysis of this work was pointed out in~\cite{Han2021_knap}. 
	%A streaming fashion model is an excellent approach to solving submodular optimization for massive data since it needs less memory and running time and reduces queries than offline algorithms. Two recent algorithms of \cite{sub_knap_stream} and  \cite{streaming-skm-25} are streaming ones with the same factor of $2.5+\epsilon$ and query complexity  $O(n\log (B)/\epsilon^4)$ in one pass and two passes, respectively where $B$ is the limited cost of the solution. The algorithm of \cite{streaming-smk-multi} reaches a better factor of $2+\epsilon$ but also requires a more significant number of passes and queries of $O(1/\epsilon)$ and $O(n\log^2(B)/\epsilon^8)$, respectively. Especially, \cite{li-linear-knap-nip22} made a breakthrough theoretical by developing a $(2+\epsilon)$-approximation algorithm in a clean linear number of queries $O(n\log(1/\epsilon)/\epsilon)$. They also showed no existing constant factor approximation in $O(n/\log n)$ queries. 
	\paragraph{$\SMK$ problem with non-monotone function.} Solving non-monotone $\SMK$ is more challenging than the monotone case. First, the property of the monotone property plays an important role in analyzing the theoretical bound. Besides, algorithms for the non-monotone case need more queries to obtain information from all elements in the condition that the marginal contribution of an element may be negative.
	
	Randomized methods are efficient tools for designing algorithms for the non-monotone $\SMK$ problem with theoretical bounds. The authors in \cite{lee-car-10}  first introduced a randomized algorithm with a factor of $5+ \epsilon$ in a polynomial time; the factor was later enhanced to $4+\epsilon$ by \cite{Kulik13_Submax_knap}. Since then, several works tried to enhance the approximation factor to  $e+\epsilon$  \cite{ChekuriVZ14,ran-smk-Felman11,Ene_HuyLNg_knap,BuchbinderF19-bestfactor}. The best factor in this research line is due to \cite{BuchbinderF19-bestfactor}. In the seminal work, they introduced a randomized algorithm with a $2.6$ factor by combining the multi-linear extension method and the rounding scheme technique in \cite{Kulik13_Submax_knap}.
	However, this work required an expensive query complexity to handle multi-linear extensions.  Authors in
	\cite{Amanatidis2021a_knap} first proposed a fast sample greedy algorithm  with $5.83+\epsilon$ factor in nearly linear query complexity of $O(n \log(n/\epsilon)/\epsilon)$.
	The factor was improved to $4+\epsilon$ by Han {\em et al.}~\cite{Han2021_knap}. More recently, Pham {\em et al.}~\cite{pham-ijcai23} significantly reduced the query complexity to linear while keeping the $4+\epsilon$ factor.
	
	Deterministic algorithms often give better results in practice than randomized ones \cite{Kuhnle19-nip,pham-ijcai23,Han2021_knap}. The first work in the research line was due to \cite{Gupta_nonmono_constrained_submax}. The authors first presented a deterministic algorithm with a factor of $6$.  However, it took $O(n^5)$ query complexity. Since then, there are several algorithms have been proposed to reduce the number of queries. The FANTOM algorithm \cite{Mirzasoleiman_data_sum} reduced queries to $O(n^2 \log(n)/\epsilon)$ but had a larger factor of $10$. The authors in \cite{nearly-liner-der-2018} introduced an algorithm with $9.5+ \epsilon$ factor in $O(nk)\max\{\epsilon^{-1}, \log \log n\}$ and it can be adapted for richer constants. \cite{best-dla} achieved a best factor of $4+\epsilon$ for a deterministic algorithm. However, the high query complexity of $O(n^3 \log(n/\epsilon)/\epsilon)$ made it may be impractical. Recently, the authors in \cite{Han2021_knap} introduced a fast deterministic algorithm with $6+\epsilon$ factor within nearly-linear queries $O(n\log (k/\epsilon)/\epsilon)$. Currently, the faster deterministic algorithm is due to Pham {\em et al.}~\cite{pham-ijcai23}. They introduced an algorithm named $\DLA$ that returned a factor of $6+\epsilon$ and wasted $O(n)$ queries. Our algorithm has improved the factor to $5+\epsilon$ in linear query complexity by contributing some important changes with tighter analysis.
	\section{Preliminaries}
	\label{sec:preli}
	A set function $f: 2^V \mapsto \mathbb{R}_+ $, defined on all subsets of a ground set $V$ of size $n$  is  submodular iff  for any $A \subseteq B \subseteq V$ and $e \notin B$, we have: 
	\begin{align*}
	f(A \cup \{e\}) - f(A) \geq f(B \cup \{e\})- f(B).
	\end{align*}
	Each element $e \in V$ has a positive cost $c(e)>0$, and the total cost of a set $S\subseteq V$ is a modular function, i.e. $c(S) =\sum_{e \in S}c(e)$. Given a positive number $B$ (knapsack constraint), we assume that every item $e\in V$ satisfies $c(e) \leq  B$; otherwise, we can simply discard it. The Submodular Maximization under Knapsack ($\SMK$) is formally defined as follows:
	\begin{definition}[$\SMK$ problem] Given a ground set $V$,  a submodular function  $f: 2^V \mapsto \mathbb{R}_+ $,
		and a positive number $B$, The $\SMK$ problem aims at finding a subset $S$ with the total cost $c(S)\leq B$ so that $f(S)$ is maximized.
	\end{definition}
	An instance of $\SMK$ is denoted by a tuple $(f, V, B)$. For simplicity, we assume that $f$ is non-negative, i.e. $f(X)\geq 0$ for all $X\subseteq V$ and  normalized, i.e., $f(\emptyset)=0$. We define the contribution gain of a set $S$ to a set $T\subseteq V$ as $f(S|T)=f(S\cup T)-f(T)$, and simplify $f(\{e\}|T)$ by $f(e|S)$.
	We assume that there exists an oracle query, which when queried with the
	set $T$ returns the value of $f(T)$.
	We denote $O$ as an optimal solution with the optimal value $\opt=f(O)$, $r=\arg\max_{o \in O}c(o)$ and $O'=O\setminus\{r\}$.
	\section{The proposed Algorithm: $\EDL$}
	\label{sec:algs}
	We now introduce our \textbf{E}nhanced \textbf{D}eterministic with \textbf{L}inear query complexity ($\EDL$) algorithm that has an approximation factor of $5+\epsilon$ in the query complexity of $O(n\log (1/\epsilon)/\epsilon)$ for $\SMK$ problem. 
	
	\subsection{Algorithm description}
	$\EDL$ receives an instance $(f, V, B)$ and a parameter $\epsilon>0$ as inputs. $\EDL$ works in two steps as follows:
	\\
\textbf{Step 1:} It calls a subroutine, the $\LA$ algorithm in \cite{pham-ijcai23}, to get a feasible solution $S'$ with an approximation factor of $19$. Therefore, it provides a $\log_{\frac{1}{1-\epsilon'}}(\frac{19}{\epsilon'^2}) \rceil+1=O(\log(19/\epsilon')/\epsilon')$ guesses of $\opt$ (recall that $\opt$ is the value of an optimal solution) in a range of $[M, 19M]$ where $M=f(S')$ (line 1),  where $\epsilon'=\epsilon/14$. This step helps to value of  threshold   $\theta=\frac{(1-\epsilon')^i 19M}{5\epsilon' B}$ may vary from $\frac{\opt}{5\epsilon'B}$ down to $\frac{\epsilon'(1-\epsilon')\opt}{B}$, a reasonable rage to get the  desired factor.
%It must be exist a guess $v=(1-\epsilon)^i \in  [M, 19M]$ so that $v\leq \opt \leq v/(1-\epsilon)$. 
\\
\textbf{Step 2:} This step  works in a main loop with $O(\log(19/\epsilon')/\epsilon')$ iterations (Lines 3-9). At each iteration $i$, it adapts the greedy threshold to add elements with high-density gain (i.e. the ratio between the marginal gain of an element and its cost) into two disjoint sets $X$ and $Y$. Elements are subsequently added to the set $T \in \{X, Y\}$ to which has the larger density gain without violating the budget constraint, as long as the density gain is at least $\theta=\frac{19M(1-\epsilon')^i}{5\epsilon' B}$ (Line 6).
	Finally, the algorithm returns the best solution between $X$ and $Y$. The details of $\EDL$ are presented in Algorithm~\ref{alg:dla}.
	\begin{algorithm}
		\SetNlSty{text}{}{:}
		%\SetAlgoNlRelativeSize{0}
		\KwIn{An instance $(f, V, B)$, $\epsilon$.} 
		%	\KwOut{A solution $S$}
		\tcp{Step 1: Pre-processing}
		$S'\leftarrow$ Result of  $\LA$~\cite{pham-ijcai23} with inputs $(f, V, B)$, $M \leftarrow f(S'), \epsilon' \leftarrow \epsilon/14$ \label{dla:p1-b}
		\\
			\tcp{Step 2: Construct candidate solutions}
		$X\leftarrow \emptyset, Y \leftarrow \emptyset$%, Z\leftarrow \emptyset$
		\\
		\For{$i=0$ to $\lceil \log_{\frac{1}{1-\epsilon'}}(\frac{19}{\epsilon'^2}) \rceil+1$}
		{
			$\theta \leftarrow 19M (1-\epsilon')^i/(5\epsilon' B)$
			\\
			\ForEach{$e \in V \setminus (X\cup Y)$}
			{
					Find $T \in \{X, Y\}$ with $c(T \cup \{e\})\leq B$ such that:
					$T = \arg \max_{T \in \{X, Y\}, \frac{f(e|T)}{c(e)}\geq\theta}\frac{f(e|T)}{c(e)}$
					\\
					\textbf{ If} such set $T$ exists  \textbf{then} $T \leftarrow T \cup \{e\}$
			}
		}
		\label{dla:p1-e}
		%	\For{$l=0$ to $ \Delta$}
		%	{ \label{dla:p2-b}
		%		$X'_{(l)} \leftarrow \arg\max_{X^i: c(X^i)\leq \epsilon'B(1+\epsilon')^l, i\leq |X|}i$
		%		\\
		%		$Y'_{(l)} \leftarrow \arg\max_{Y^i: c(Y^i)\leq \epsilon'B(1+\epsilon')^l, i\leq |X|}i$
		%		\\
		%		$e_X \leftarrow  \arg\max_{e \in V:  c(X'_{(l)}\cup \{e\})\leq B}f(X'_{(l)}\cup \{e\})$ \label{alg:dla-select_ex}
		%		\\
		%		$e_Y \leftarrow  \arg\max_{e \in V:  c(Y'_{(l)}\cup \{e\})\leq B}f(Y'_{(l)}\cup \{e\})$
		%		\\
		%		$X_{(l)}\leftarrow X'_{(l)} \cup \{e_X\}$, 
		%		$Y_{(l)}\leftarrow Y'_{(l)} \cup \{e_Y\}$
		%	} \label{dla:p2-e}
		$S \leftarrow \arg \max_{T \in \{ X, Y \}}f(T)$
		\\ 
		\Return $S$.
		\caption{$\EDL$ Algorithm}
		\label{alg:dla}
	\end{algorithm}
	%The second phase (lines~\ref{dla:p2-b}-\ref{dla:p2-e}) is to improve the quality of candidate solution $T\in \{X, Y\}$ which was obtained at the end of phase 1.  Denote  $T^i$ as a set of the first $i^{th}$ elements added in $T$. 	Our main observation is that the performance of $\DLA$ depends on the cost of $T'=\arg\max_{T^i: c(T^i)\leq B-c(r)}(i)$. Recall that $r$ is $\arg\max_{o\in O}c(o)$ and $c(r)\leq B$.
	%Thus we scan an upper bound of $c(T')$ from $\epsilon'B$ to $B$ and improve the quality of $T'$ by adding into it an element $e=\arg\max_{e\in V: c(T')+c(e)\leq B}f(T'\cup \{e\})$ (lines 13-15).
	\paragraph{Theoretical analysis.}
%	The following Lemmas give the bounds of the final solution when $c(r)< (1-\epsilon')B$ and $c(r)\geq(1-\epsilon')B$, respectively.
We first provide the following useful notations.
	\begin{itemize}
		\item $O'=O\setminus \{r\}$.
		\item 	Assuming that $X=\{x_1, x_2, \ldots x_{|X|}\}$ we denote $X^i=\{x_1, x_2, \ldots x_{i}\}$, and $t=\max\{i: c(X^i)+ c(r)\leq B\}$.
		\item  $Y=\{y_1, y_2, \ldots y_{|Y|}\}$ we denote $Y^i=\{y_1, y_2, \ldots y_{i}\}$, and $u=\max\{i: c(Y^i)+ c(r)\leq B\}$.
		\item $X_j$ and $Y_j$ are $X$ and $Y$ after iteration $j$ of the first loop of the Algorithm \ref{alg:dla}, respectively.
		\item 	
		For $e\in X \cup Y$, we denote by $X^{<e}$ and $Y^{<e}$ the set of elements in $X$ and $Y$ before adding $e$ into  $X$ or $Y$, respectively and	$l(e)$ is the iteration when $e$ is added to $X$ or $Y$.
		\item Denote by $\theta_i$ as $\theta$ at the iteration $i$, by $\theta^X_{(j)}$ as $\theta$ when the element $x_j$  is added into $X$, and by $\theta^Y_{(j)}$ as $\theta$ when the element $y_j$  is added into $Y$. 
		\item Finally,  $\theta_{last}$ is $\theta$ at the last iteration of the first loop.
	\end{itemize}
	We first provide a bound of the optimal solution when $c(r)<(1-\epsilon')B$ by carefully analyzing the cost of $X$ and $Y$ in comparison with $B-c(r)$ in Lemma~\ref{lem:r-small}.
	\begin{lemma}If $c(r)< (1-\epsilon')B$, we have $f(O)\leq (5+\epsilon)f(S)$.
		\label{lem:r-small}
	\end{lemma}
	\begin{proof} We prove the Lemma by considering the following cases:
		\\
		\textbf{Case 1.} If both $c(X)\geq B-c(r)$ and $c(Y)\geq B-c(r)$. $X$ must contain at least $t+1$ elements and $Y$ must contain at least $u+1$ elements.
		 Without loss of generality, we assume that the algorithm obtains $Y^{u+1}$ before $X^{t+1}$. If the algorithm obtains $X^{t+1}$ after the first iteration, we have $f(S)\geq f(X^{t+1})\geq c(X^{t+1})\theta_{1}\geq \frac{\opt}{5}$. Thus the Lemma holds. We consider the otherwise case, i.e., $X^{t+1}$ obtained at the iteration $j\geq 2$. Denote $Y^q=Y^{<x_{t+1}}$ we get $Y^{u+1}\subseteq Y^q$.  By the selection rule of the algorithm, each element $e \in Y^{q}$ has the density gain satisfying:
		\begin{align}
		\frac{f(e|X^{t+1})}{c(e)} \leq 	\frac{f(e|X^{<e})}{c(e)} \leq  \frac{f(e|Y^{<e})}{c(e)}.
		\label{lem:dla1-rule}
		\end{align}	
		By the submodularity of $f$, we have:
		\begin{align}
		f(O'\cup X^{t+1})-f(X^{t+1}) &\leq \sum_{e\in O' \setminus X^{t+1}}f(e|X^{t+1})
		\\
		&=  \sum_{e\in O' \cap Y^q}f(e|X^{t+1})+  \sum_{e\in O' \setminus (X^{t+1}\cup Y^q)}f(e|X^{t+1}) 
		\\
		& \leq  \sum_{e\in O' \cap Y^q}f(e|Y^{<e})+  \sum_{e\in O' \setminus (X^{t+1}\cup Y^q)}f(e|X^{t+1}) 
		\label{ine:lem7.11}
		\end{align}
		Similarly, we also get
		\begin{align}
		f(O'\cup Y^q)-f(Y^q)& \leq \sum_{e\in O' \setminus Y^q}f(e|Y^q)
		\\
		& =\sum_{e\in O' \cap X^{t+1}}f(e|Y^{q})+ 
		\sum_{e\in O' \setminus (X^{t+1}\cup Y^q)}f(e|Y^q)
		\\
		&\leq  \sum_{e\in O' \cap X^{t+1}}f(e|X^{<e})+  \sum_{e\in O' \setminus (X^{t+1}\cup Y^q)}f(e|Y^q)
		\label{ine:lem7.12}
		\end{align}
		Combine inequalities \eqref{ine:lem7.11} and \eqref{ine:lem7.12}, we obtain
		\begin{align}
		&f(O'\cup X^{t+1})-f(X^{t+1})+ f(O'\cup Y^q)-f(Y^q) \label{ine:lem7:c11}
		\\
		&\leq  \sum_{e\in O' \cap Y^q}f(e|Y^{<e})+  \sum_{e\in O' \setminus (X^{t+1}\cup Y^q)}f(e|X^{t+1})  \nonumber
		\\
		& + \sum_{e\in O' \cap X^{t+1}}f(e|X^{<e})+  \sum_{e\in O' \setminus (X^{t+1}\cup Y^q)}f(e|Y^q) \label{ine}
		\end{align}
		We bound the right hand side of \eqref{ine} by finding the connection between $O' \setminus (X^{t+1}\cup Y^q)$ and $X^{t+1}\setminus O'$. The density gain of any element $e\in  O'\setminus (X^{t+1}\cup Y^q)$ is less than the threshold at the  $l(x_{t+1})$-th iteration, i.e., $f(e|X^{t+1})\leq \frac{\theta^X_{(t+1)}}{1-\epsilon'}$ which implies
		\begin{align}
		\sum_{e\in O'\setminus (X^{t+1}\cup Y^q)}f(e|X^{t+1}) &\leq \sum_{e\in O'\setminus (X^{t+1}\cup Y^q)}f(e|X^{<e})
		\\
		&\leq c(O'\setminus (X^{t+1}\cup Y^q))\frac{\theta^X_{(t+1)}}{1-\epsilon'}.
		\\
		\Longrightarrow \theta^X_{(t+1)} & \geq  \frac{(1-\epsilon')\sum_{e\in O'\setminus (X^{t+1}\cup Y^q)}f(e|X^{t+1})}{c(O'\setminus (X^{t+1}\cup Y^q))}.
		\end{align}
		Therefore 
		\begin{align}
		\sum_{e \in X^{t+1}\setminus O'}f(e|X^{<e})&\geq c(X^{t+1}\setminus O')\theta^X_{(t+1)}
		\\
		&\geq(1-\epsilon')\frac{c(X^{t+1}\setminus  O')}{ c(O'\setminus (X^{t+1}\cup Y^q))}	\sum_{e\in O'\setminus (X^{t+1}\cup Y^q)}f(e|X^{t+1})
		\\
		&\geq  (1-\epsilon')\sum_{e\in O'\setminus (X^{t+1}\cup Y^q)}f(e|X^{t+1}) \label{lem7:c1}
		\end{align}
		where the lase inequality is due to the fact that $c(X^{t+1})>c(O')$ thus $c(X^{t+1}\setminus O')\geq c(O'\setminus X^{t+1})\geq c(O'\setminus (X^{t+1}\cup Y^q))$.
	Applying the similar transform from \eqref{ine:lem7:c11} to \eqref{lem7:c1} with note that $Y^{u+1} \subseteq Y^q$, we have
		\begin{align}
		\sum_{e\in O'\setminus (X^{t+1}\cup Y^q)}f(e|Y^{q})&\leq \sum_{e\in O'\setminus (X^{<y_{u+1}}\cup Y^{u+1})}f(e|Y^{<e})
		\\
		&\leq c(O'\setminus (X^{<y_{u+1}}\cup Y^{u+1}))\frac{\theta^Y_{(u+1)}}{1-\epsilon'}.
		\end{align}
		It follows that
		\begin{align}
		\sum_{e \in Y^{q}\setminus O'}f(e|Y^{<e}) &\geq 	\sum_{e \in Y^{u+1}\setminus O'}f(e|Y^{<e})\geq  c(Y^{u+1}\setminus O')\theta^Y_{(u+1)}
		\\
		&\geq (1-\epsilon') \frac{c(Y^{u+1}\setminus O') \sum_{e\in O'\setminus (X^{t+1}\cup Y^q)}f(e|X^{t+1})}{c(O'\setminus (X^{<y_{u+1}}\cup Y^{u+1}))}  
		\\
		&
		\geq  (1-\epsilon') \sum_{e\in O'\setminus (X^{t+1}\cup Y^q)}f(e|X^{t+1}).  \label{lem7:c2}
		\end{align}
		Put  \eqref{lem7:c1} and \eqref{lem7:c2} into \eqref{ine}, we have
		\begin{align}
		&f(O'\cup X^{t+1})-f(X^{t+1})+ f(O'\cup Y^q)-f(Y^q) 
		\\
		&\leq 
		\sum_{e\in O' \cap Y^q}f(e|Y^{<e})+  	\frac{1}{1-\epsilon'}\sum_{e \in Y^{q}\setminus O'}f(e|Y^{<e})  
		\\
		&+\sum_{e\in O' \cap X^{t+1}}f(e|X^{<e})+ \frac{1}{1-\epsilon'}\sum_{e \in X^{t+1}\setminus O'}f(e|X^{<e})
		\\
		& < \frac{1}{1-\epsilon'}\sum_{e \in Y^{q}}f(e|Y^{<e})  + \frac{1}{1-\epsilon'}\sum_{e \in X^{t+1}}f(e|X^{<e})
		\\
		&\leq  \frac{f(X^{t+1})+f(Y^q)}{1-\epsilon'}.
		\end{align}
		The algorithm always selects elements with no-negative marginal gain into $X$, $Y$, so we have: $f(X^i)\leq f(X)$ and $f(Y^j)\leq f(Y)$ for $i=1, \ldots, |X|$ and $j=1,\ldots, |Y|$.
		Combine this with the submodularity of $f$, we get
		\begin{align}
		f(O)&\leq f(O')+f(r)\leq f(O'\cup X^{t+1})+f(O'\cup Y^q)+f(r)
		\\
		&< (1+\frac{1}{1-\epsilon'})(f(X^{t+1})+f(Y^q))+f(r)
		\\
		&\leq (5+\frac{2\epsilon'}{1-\epsilon'})f(S).
		\\
		&\leq (5+\epsilon )f(S).
		\end{align}
		\textbf{Case 2.} If $c(X)$ or $c(Y)$ is greater than or equal to $B-c(r)$ and the rest is less than $B-c(r)$. Without loss of generality, we assume that $c(X)\geq  B-c(r)$ and $c(Y)<B-c(r)$. 	Each element $e \in O\setminus Y$ has the  density gain with $Y$ is less than $\theta_{last}$, i.e.,
		$
		\frac{f(e|Y)}{c(e)}< \theta_{last}
		$, 
		so we get:
		\begin{align}
		f(Y\cup O)-f(Y)&\leq \sum_{e\in O\setminus Y}f(e|Y)
		\\
		& \leq   \sum_{e\in O\cap X}f(e|Y)+\sum_{e\in O\setminus (X\cup Y)}f(e|Y)
		\\
		& < f(X)+c(O)\frac{\epsilon'(1-\epsilon')\opt}{B}
		\\
		&\leq f(X)+\epsilon'(1-\epsilon')\opt.
		\end{align}
	Elements are selected into $X^{t+1}$ have the  density gain at least $\theta^X_{(t+1)}$, so we have:
		\begin{align}
		\sum_{e\in X^{t+1}\setminus O'}f(e|X^{<e})\geq \sum_{e\in X^{t+1}\setminus O'} c(e)\theta^X_{(t+1)}\geq c(X^{t+1}\setminus O')\theta^X_{(t+1)}
		\end{align}
		Therefore
		\begin{align}
		f(O'\cup X^{t+1})-f(X^{t+1})&\leq \sum_{e\in O' \cap Y^q}f(e|Y^{<e})+  \sum_{e\in O' \setminus (X^{t+1}\cup Y^q)}f(e|X^{t+1}) 
		\\
		& \leq f(Y)+c(O'\setminus (X^{t+1}\cup Y^q))\frac{\theta_{(t+1)}}{1-\epsilon'}
		\\
		&\leq f(Y)+\frac{c(O'\setminus (X^{t+1}\cup Y^q)) \sum_{e\in X^{t+1}\setminus O'} f(e|X^{<e})}{(1-\epsilon')c(X^{t+1}\setminus O')}
		\\
		& \leq f(Y)+\frac{f(X^{t+1})}{1-\epsilon'}
		\\
	&	\leq f(Y)+\frac{f(X)}{1-\epsilon'}
		\end{align}
	It follows that
		\begin{align}
		f(O)&\leq f(O\cup X^{t+1})+f(O\cup Y)
		\\
		&\leq  f(O'\cup X^{t+1})+f(r)+f(O\cup Y)
		\\
		& \leq 2f(Y)+f(X)+\frac{f(X^{t+1})}{1-\epsilon'} + f(r)+ \epsilon'(1-\epsilon')\opt
		\\
		&\leq \frac{5-4 \epsilon'}{1-\epsilon'}f(S)+ \epsilon'(1-\epsilon')\opt.
		\end{align}
		which implies that 
		\begin{align}
		f(O) &\leq \frac{5-4\epsilon'}{(1-\epsilon'(1-\epsilon'))(1-\epsilon')}f(S)
		\\
		& \leq \frac{5-4\epsilon'}{(1-\epsilon')^2}f(S) \leq (5+\frac{6-5\epsilon'}{(1-\epsilon')^2}\epsilon')f(S)
		\\ 
		&  
		\leq (5+\epsilon)f(S).
		\end{align}
		\textbf{Case 3.} Both $c(X)$ and $c(Y)$ are less than $B-c(r)$, we have:
		\begin{align}
		f(O\cup X)-f(X)+f(O\cup Y)-f(Y)&=\sum_{e\in O\setminus X}f(e|X)+\sum_{e\in O\setminus Y}f(e|Y)
		\\
		&< f(X)+f(Y)+ 2 c(O)\frac{\epsilon'(1-\epsilon')\opt}{B}
		\\
		&\leq 2f(S)+ 2\epsilon'(1-\epsilon')\opt
		\end{align}
		which implies that 
		$
		f(O)\leq f(O\cup X)+f(O\cup Y) \leq 4f(S)+2\epsilon'(1-\epsilon')\opt
		$. Hence $f(O)\leq \frac{4f(S)}{1-2\epsilon'(1-\epsilon')}<5f(S)$.
		Combining all cases, we complete the proof.
	\end{proof}
	On the remain case $c(r) \geq (1-\epsilon')B$, we find another connection between $O$ and $X$, $Y$ by consider some iteration $j$ so that $\frac{(1-\epsilon')\opt}{5  B}\leq \theta_j< \frac{\opt}{5 B}$.
	%%%% Check here
	\begin{lemma}
		If $c(r)\geq (1-\epsilon')B$, one of three things happens:
		\begin{itemize}
			\item[a)] $f(S)\geq\frac{(1-\epsilon')^2}{5}\opt$. 
			\item[b)] There exists a subset $X' \subseteq X$ so that $f(O\cup X')\leq 2f(S)+ \epsilon'(1-\epsilon)\opt$.
			\item[c)] There exists a subset $X' \subseteq X$ so that $f(O'\cup X')\leq 2f(S)+ \epsilon'\opt/5$.
		\end{itemize}
		Similarly,  one of three things happens:
		\begin{itemize}
			\item[d)] $f(S)\geq\frac{(1-\epsilon')^2}{5}\opt$.
			\item[e)] There exists a subset $Y' \subseteq Y$ so that $f(O\cup Y')\leq 2f(S)+ \epsilon'(1-\epsilon')\opt$.
			\item[f)] There exists a subset $Y' \subseteq Y$ so that $f(O'\cup Y')\leq 2f(S)+ \epsilon'\opt/5$.
		\end{itemize}
		\label{lem:dla2}
	\end{lemma}
	\begin{proof}
		In this case, we have $c(O\setminus \{r\})\leq \epsilon' B, c(X^t)\leq \epsilon' B$. Since the roles of $X$ and $Y$ are the same, we first consider analyzing the cases for $X$. We then derive similar outcomes for $Y$.
		\\
		\textbf{Case 1.} If $X^t$ is $X$ after ending the first loop, we have
		\begin{align}
		f(O\cup X^t)-f(X^t)&\leq \sum_{e\in O\setminus X^t} f(e|X^t)
		\\
		&\leq \sum_{e \in O\cap Y^{<x_t}} f(e|X^t) +\sum_{e O\setminus(X^t\cup Y^{<x_t})} f(e|X^t)
		\\
			&\leq \sum_{e \in O\cap Y^{<x_t}} f(e|Y^{<e}) +\sum_{e O\setminus(X^t\cup Y^{<x_t})} f(e|X^t)
		\\
		& < f(Y)+\epsilon'(1-\epsilon')\opt
		\end{align}
		Thus $	f(O\cup X^t)\leq f(X^t)+f(Y)+\epsilon'(1-\epsilon')\opt\leq 2f(S)+\epsilon'(1-\epsilon')\opt $.
		\\
		\textbf{Case 2.} If $X^t\subset X$, $X$ contains at least $t+1$ elements. 
		There exist an iteration $j$, at which we have:
		\begin{align*}
		\frac{(1-\epsilon')\opt}{5  B}\leq \theta_j=\frac{19\Gamma (1-\epsilon')^j}{5 \epsilon' B} < \frac{\opt}{5 B}.
		\end{align*}
		We further consider two following sub-cases:
		\\
		\textbf{Case 2.1.} If $X^{t+1} \subseteq X_j$.  If $c(X_j)\geq (1-\epsilon')B$, then $$f(S)\geq f(X_j)\geq c(X_j)\theta_j\geq \frac{(1-\epsilon')^2}{5}\opt.$$ 
		If $c(X_j)< (1-\epsilon')B$.
		Since $c(X_j)+\max_{e \in O\setminus \{r\}}c(e)\leq c(X_j)+ c(O\setminus\{r\})<B$, we get $\frac{f(e|X_j)}{c(e)}< \theta_j$ for any  $e \in O'\setminus(X_j\cup Y_j)$. Therefore:
		\begin{align*}
		f(X_j\cup O')-f(X_j)&\leq   \sum_{e \in O'\setminus X_j}f(e|X_j)
		\\
		& =\sum_{e \in O'\cap Y_j}f(e|X_j)+\sum_{e \in O'\setminus (X_j\cup Y_j)}f(e|X_j)
		\\
		&<f(Y_j)+ \sum_{e \in O'\setminus (X_j\cup Y_j)}c(e)\theta_j
		\\
		& < f(Y) + \epsilon' B \frac{\opt}{5 B}=f(Y)+\frac{\epsilon' \opt}{5}.
		\end{align*}
		Therefore
		\begin{align*}
		f(X_j\cup O') & \leq 
		2f(S)+\frac{\epsilon'\opt}{5} .
		\end{align*}
		\textbf{Case 2.2.} If $ X_j\subset X^{t+1}$. For any element $e \in V\setminus (X^{t}\cup Y^{<x_{t}})$, its  density gain concerning $X^{t}$ is smaller than the threshold at the previous iteration (in the first loop), thus
		\begin{align}
		\frac{f(e|X^{t})}{c(e)} < \frac{\theta_{(t+1)}}{1-\epsilon'}\leq \theta_j<\frac{\opt}{5 B}.
		\end{align}
		With notice that $c(O')< \epsilon' B$, we get
		\begin{align*}
		f(X^t\cup O')-f(X^t )&\leq\sum_{e \in O'\setminus X^t}f(e|X^t)
		\\
		& =\sum_{e \in O'\cap Y^t}f(e|X^t)+ \sum_{e \in O'\setminus  (X^t\cup Y^t)}f(e|X^t)
		\\
		& < \sum_{e \in O'\cap Y^t}f(e|Y^{<e})+ \sum_{e \in O'\setminus  (X^t\cup Y^t)}c(e) \theta_j
		\\
		& \leq f(Y^{<x_t})+ c(O') \frac{\opt}{5 B}
		\\
		& \leq f(Y)+ \frac{\epsilon' \opt}{5}.
		\end{align*}
		which implies that 
		\begin{align*}
		f(X^t \cup O') \leq 2f(S)+ \frac{\epsilon'\opt}{5}.
		\end{align*}
		Combining all the above cases, we get the proof for \textit{a)}, \textit{b)}, \textit{c)}.
		By the similarity argument for $Y'\subseteq Y$, we also have the proof for \textit{d)}, \textit{e)}, \textit{f)}.
	\end{proof}
	Finally, using Lemma~\ref{lem:r-small} and Lemma~\ref{lem:dla2}, we get the performance guarantee's $\EDL$.
	\begin{theorem}
		For any $\epsilon \in (0, 1)$,	$\EDL$ is a deterministic algorithm that has a query complexity $O(n\log(1/\epsilon)/\epsilon)$ and returns an approximation factor of $5+ \epsilon$.
		\label{theo:dla}
	\end{theorem}
	\begin{proof}
		$\DLA$ fist calls $\LA$ in \cite{pham-ijcai23} to find a feasible solution $S'$. This task takes $O(n)$ queries (Theorem 1 in \cite{pham-ijcai23}). It then consists of $O(\log(1/\epsilon')/\epsilon')$  iterations. Each iteration of these loops takes $O(n)$ queries; thus we get the total number of queries at most:
		$$
		O(n)+ n \cdot O(\frac{1}{\epsilon'}\log(\frac{1}{\epsilon'}))=O(\frac{n}{\epsilon}\log(\frac{1}{\epsilon})). 
		$$
		To prove the approximation factor, we consider the following cases:
		If $c(r)<(1-\epsilon')B$, the get the approximation factor due to Lemma~\ref{lem:r-small}, so we now consider the case $c(r)\geq (1-\epsilon')B$ by using Lemma \ref{lem:dla2}. If one of \textbf{a)} \textbf{or} \textbf{d)} happens. Since
		$\epsilon'=\frac{\epsilon}{14}<\frac{1}{14}$ we get: 
		$$
		\opt \leq \frac{5f(S)}{(1-\epsilon')^2}\leq 5(1+\frac{14}{13}\epsilon')^2 f(S) <(5+\epsilon)f(S)
		$$ thus the Theorem holds. Since the roles of $X$ and $Y$ are the same, we will consider the following cases: both \textbf{b)} and \textbf{e)} happen and both \textbf{b)} and \textbf{f)} happen
		\\
		\textbf{Case 1.} If both \textbf{b)} and \textbf{e)} happen. Applying Lemma~\ref{lem:dla2} we have
		\begin{align}
		f(O) &\leq f(O\cup X)+ f(O\cup Y) 
		\leq  4 f(S)+ 2\epsilon'(1-\epsilon)\opt
		\end{align}
		which implies that $f(O)\leq \frac{4f(S)}{1-2\epsilon'(1-\epsilon')}<5f(S)$.
		\\
		\textbf{Case 2.} If both \textbf{b)} and \textbf{f)} happen, we have
		\begin{align}
		f(O)&\leq f(O\cup X')+f(O\cup Y')
		\\
		&\leq f(O\cup X')+ f(O'\cup Y')+ f(r)
		\\
		& \leq 5f(S)+ (\epsilon'(1-\epsilon')+\frac{\epsilon'}{5})\opt
		\end{align}
		which implies that  $$f(O)\leq \frac{5f(S)}{1-(\epsilon'(1-\epsilon')+\frac{\epsilon'}{5})}\leq \frac{5}{1-2\epsilon'} f(S) \leq (5+\epsilon)f(S).$$
		Combining two cases, we obtain the proof.
	\end{proof}
	\section{Conclusions}
	In this work, we have proposed an approximation algorithm for the well-known Submodular Maximization under Knapsack constraint. Our algorithm keeps the best query complexity of $O(n)$ while improving the approximation factor from $6+\epsilon$ to $5+\epsilon$. The key technique of our algorithm is to optimize steps of the fastest algorithm in \cite{pham-ijcai23} with a tighter theoretical analysis.  In the future, we will address another valuable question: can we further improve the factor of linear query complexity in a deterministic algorithm for the $\SMK$ problem?
	\label{sec:con}
	\bibliographystyle{spmpsci} 
	\bibliography{SMK-ref}
\end{document}